\documentclass[11pt]{article}
\usepackage[english]{babel}
\usepackage{verbatim}
\usepackage{authblk}
\usepackage{amsmath, amsthm, amssymb, graphics, graphicx}
\DeclareMathOperator{\amp}{amp}

\newcommand{\mathsym}[1]{{}}
\newtheorem{theorem}{Theorem}[section]
\newtheorem{lemma}{Lemma}[section]
\newtheorem{claim}{Claim}[section]

\newtheorem{conjecture}{Conjecture}[section]

\title{\bf\Large{{Chromatic number of ISK4-free graphs}}}
\author{Ngoc Khang Le 
\thanks{This work was performed within the framework of the LABEX MILYON (ANR-10-LABX-0070) of Universit\'e de Lyon, within the program ``Investissements d'Avenir'' (ANR-11- IDEX-0007) operated by the French National Research Agency (ANR). Partially supported by ANR project Stint under reference ANR-13-BS02-0007.}
}
\affil{LIP, ENS de Lyon, Lyon, France}

\begin{document}
\maketitle

\begin{abstract}
A graph $G$ is said to be ISK4-free if it does not contain any subdivision of $K_4$ as an induced subgraph. In this paper, we propose new upper bounds for chromatic number of  ISK4-free graphs and $\{$ISK4, triangle$\}$-free graphs. 
\end{abstract}

\section{Introduction}

We say that a graph $G$ is \emph{$H$-free} if $G$ does not contain any induced subgraph isomorphic to $H$. For $n\geq 1$, denote by $K_n$ the complete graph on $n$ vertices. A \textit{subdivision} of a graph $G$ is obtained by subdividing its edges into paths of arbitrary length (at least one). We say that $H$ is \textit{an ISK4 of a graph} $G$ if $H$ is an induced subgraph of $G$ and $H$ is a subdivision of $K_4$. A graph that does not contain any induced subdivision of $K_4$ is said to be \textit{ISK4-free}. For instance, series-parallel graphs and line graph of cubic graphs are ISK4-free (see \cite{LMT12}). A \textit{triangle} is a graph isomorphic to $K_3$. 

The \textit{chromatic number} of a graph $G$, denoted by $\chi(G)$, is the smallest integer $k$ such that $G$ can be partitioned into $k$ stable sets. Denote by $\omega(G)$ the size of a largest clique in $G$. A class of graphs $\cal G$ is \textit{$\chi$-bounded} with \textit{$\chi$-bounding function $f$} if, for every graph $G\in {\cal G}$, $\chi(G)\leq f(\omega(G))$. This concept was introduced by Gy\'arf\'as \cite{G87} as a natural extension of perfect graphs, that form a $\chi$-bounded class of graphs with $\chi$-bounding function $f(x)=x$. The question is: which induced subgraphs need to be forbidden to get a $\chi$-bounded class of graphs? One way to forbid induced structures is the following: fix a graph $H$, and forbid every induced subdivision of $H$. We denote by \emph{Forb$^*(H)$} the class of graphs that does not contain any induced subdivision of $H$. The class Forb$^*(H)$ has been proved to be $\chi$-bounded for a number of graph $H$. Scott \cite{S97} proved that for any forest $F$, Forb$^*(F)$ is $\chi$-bounded. In the same paper, he conjectured that Forb$^*(H)$ is $\chi$-bounded for any graph $H$. Unfortunately, this conjecture has been disproved (see \cite{PK14}). However, there is no general conjecture on which graph $H$, Forb$^*(H)$ is $\chi$-bounded. This question is discussed in \cite{CELM14}. We focus on the question when $H=K_4$. In this case, Forb$^*(K_4)$ is the class of ISK4-free graphs. Since $K_4$ is forbidden, proving that the class of ISK4-free graphs is $\chi$-bounded is equivalent to proving that there exists a constant $c$ such that for every ISK4-free graph $G$, $\chi(G)\leq c$. Remark that the existence of such constant was pointed out in \cite{LMT12} as a consequence of a result in \cite{KO04}, but it is rather large ($\geq 2^{2^{2^{25}}}$) and very far from these two conjectures:

\begin{conjecture}[L{\'e}v{\^e}que, Maffray, Trotignon 2012 \cite{LMT12}] \label{conj:1}
If $G$ is an ISK4-free graph, then $\chi(G)\leq 4$.
\end{conjecture} 

\begin{conjecture}[Trotignon, Vu{\v{s}}kovi{\'c} 2016 \cite{TV16}] \label{conj:2}
If $G$ is an $\{$ISK4, triangle$\}$-free graph, then $\chi(G)\leq 3$.
\end{conjecture}

No better upper bound is known even for the chromatic number of $\{$ISK4, triangle$\}$-free graphs. However, attempts were made toward these two conjectures. A \textit{hole} is an induced cycle on at least four vertices. For $n\geq 4$, we denote by $C_n$ the hole on $n$ vertices. A \textit{wheel} is a graph consisting of a hole $H$ and a vertex $x\notin H$ which is adjacent to at least three vertices on $H$. The \textit{girth} of a graph is the length of its smallest cycle. The optimal bound is known for the chromatic number of $\{$ISK4, wheel$\}$-free graphs and $\{$ISK4, triangle, $C_4\}$-free graphs:  

\begin{theorem}[L{\'e}v{\^e}que, Maffray, Trotignon 2012 \cite{LMT12}] \label{Thm:wheel-free}
Every $\{$ISK4, wheel$\}$-free graph is $3$-colorable.
\end{theorem}

\begin{theorem}[Trotignon, Vu{\v{s}}kovi{\'c} 2016 \cite{TV16}] \label{Thm:girth5}
Every ISK4-free graph of girth at least $5$ contains a vertex of degree at most $2$ and is $3$-colorable.
\end{theorem}

The proof of Theorems \ref{Thm:wheel-free} and \ref{Thm:girth5} relies on structural decompositions. One way to prove Conjectures \ref{conj:1} and \ref{conj:2} is to find a vertex of small degree. This approach is successfully used in \cite{TV16} to prove Theorem \ref{Thm:girth5}. Two following conjectures will immediately imply the correctness of Conjectures \ref{conj:1} and \ref{conj:2} (definitions of $K_{3,3}$, prism and $K_{2,2,2}$ are given in Section \ref{S:2}) :  

\begin{conjecture} [Trotignon '2015 \cite{T15}] \label{conj:3}
Every $\{$ISK4, $K_{3,3}$, prism, $K_{2,2,2}\}$-free graph contains a vertex of degree at most three.
\end{conjecture}

\begin{conjecture} [Trotignon, Vu{\v{s}}kovi{\'c} '2016 \cite{TV16}] \label{conj:4}
Every $\{$ISK4, $K_{3,3}$, triangle$\}$-free graph contains a vertex of degree at most two.
\end{conjecture}

However, we find a new bound for the chromatic number of ISK4-free graphs using another approach. Our main results are the following theorems:

\begin{theorem} \label{Thm:1}
Let $G$ be an $\{$ISK4, triangle$\}$-free graph. Then $\chi(G)\leq 4$.
\end{theorem}

\begin{theorem} \label{Thm:2}
Let $G$ be an ISK4-free graph. Then $\chi(G)\leq 24$.
\end{theorem}

Remark that the bounds we found are much closer to the bound of the conjectures than the known ones. The main tool that we use to prove these theorems is classical. It is often used to prove $\chi$-boundedness results relying on the layers of neighborhood. The paper is organized as follows. We first introduce some notations in Section \ref{S:2}. Sections \ref{S:3} and \ref{S:4} are devoted to the proof of Theorem \ref{Thm:1} and \ref{Thm:2}, respectively. 

\section{Preliminaries} \label{S:2}

In this section, we present some notations and useful lemmas which will be used later in our proof. Let $G(V,E)$ be a graph, we denote by $|G|$ the number of its vertices. A vertex $v$ of the graph $G$ is \textit{complete} to a set of vertices $S\subseteq V(G)\setminus v$ if $v$ is adjacent to every vertex in $S$. A graph is called \textit{complete bipartite} (resp. \textit{complete tripartite}) if its vertex set can be partitioned into two (resp. three) non-empty stable sets that are pairwise complete to each other. If these two (resp. three) sets have size $p$, $q$ (resp. $p$, $q$, $r$) then the graph is denoted by $K_{p,q}$ (resp. $K_{p,q,r}$). A complete bipartite or tripartite graph is \textit{thick} if it contains a $K_{3,3}$. Given a graph $H$, the \textit{line graph} of $H$ is the graph $L(H)$ with vertex set $E(G)$ and edge set $\{ef:e\cap f\neq \emptyset\}$. A graph $P$ on $\{x_1,\ldots,x_n\}$ is a \textit{path} if $x_ix_j\in E(P)$ iff $|i-j|=1$ (this is often referred to \textit{induced path} in literature). The \textit{length} of a path is the number of its edges. The two \textit{ends} of $P$ are $x_1$ and $x_n$. The \textit{interior} of $P$ is $\{x_2,\ldots,x_{n-1}\}$. We denote by $x_iPx_j$ the subpath of $P$ from $x_i$ to $x_j$ and denote by $P^*$ the subpath of $P$ from $x_2$ to $x_{n-1}$ ($x_2Px_{n-1}$). A path $P$ is \textit{flat} in $G$ if all the interior vertices of $P$ are of degree $2$ in $G$. When $S\subseteq V(G)$, we denote by $N(S)$ the set of neighbors of $S$ in $G\setminus S$ and denote by $G|S$ the subgraph of $G$ induced by $S$. When $K\subseteq V(G)$ and $C\subseteq V(G)\setminus K$, we denote by $N_K(C)$ the set of neighbors of $C$ in $K$, or $N_K(C)=N(C)\cap K$.  

A \textit{cutset} in a graph is a subset $S\subsetneq V(G)$ such that $G\setminus S$ is disconnected. For any $k\geq 0$, a $k$-cutset is a cutset of size $k$. A cutset $S$ is a \textit{clique cutset} if $S$ is a clique. A \textit{proper $2$-cutset} of a graph $G$ is a $2$-cutset $\{a,b\}$ such that $ab\notin E(G)$, $V(G)\setminus \{a,b\}$ can be partitioned into two non-empty sets $X$ and $Y$ so that there is no edge between $X$ and $Y$ and each of $G[X\cup\{a,b\}]$ and $G[Y\cup\{a,b\}]$ is not a path from $a$ to $b$. A \textit{prism} is a graph made of three vertex-disjoint paths $P_1 = a_1\ldots b_1$, $P_2 = a_2\ldots b_2$, $P_3 = a_3\ldots b_3$ of length at least $1$, such that $a_1a_2a_3$ and $b_1b_2b_3$ are triangles and no edges exist between the paths except these of the two triangles. Let $S=\{u_1,u_2,u_3,u_4\}$ induces a square (i.e. $C_4$) in $G$ with $u_1$, $u_2$, $u_3$, $u_4$ in this order along the square. A \textit{link} of $S$ is a path $P$ of $G$ with ends $p$, $p'$ such that either $p=p'$ and $N_S(p)=S$, or $N_S(p)=\{u_1,u_2\}$ and $N_S(p')=\{u_3,u_4\}$, or $N_S(p)=\{u_1,u_4\}$ and $N_S(p')=\{u_2,u_3\}$, and no interior vertex of $P$ has a neighbor in $S$. A \textit{rich square} is a graph $K$ that contains a square $S$ as an induced subgraph such that $K\setminus S$ has at least two components and every component of $K\setminus S$ is a link of $S$. For example, $K_{2,2,2}$ is a rich square (it is the smallest one).

We use in this paper some decomposition theorems from \cite{LMT12}: 

\begin{lemma}[see Lemma 3.3 in \cite{LMT12}] \label{Lm:K33}
Let $G$ be an ISK4-free graph that contains $K_{3,3}$. Then either $G$ is a thick complete bipartite or complete tripartite graph, or $G$ has a clique cutset of size at most $3$.
\end{lemma}

\begin{lemma}[see Lemmas 6.1 and 7.2 in \cite{LMT12}] \label{Lm:prism,K222}
Let $G$ be an ISK4-free graph that contains a rich square or a prism. Then either $G$ is the line graph of a graph with maximum degree $3$, or $G$ is a rich square, or $G$ has a clique cutset of size at most $3$ or $G$ has a proper $2$-cutset.
\end{lemma}

\textit{Reducing} a flat path $P$ of length at least $2$ means deleting its interior and add an edge between its two ends. The following lemma shows that a graph remains ISK4-free after reducing a flat path: 

\begin{lemma} [see Lemma 11.1 in \cite{LMT12}] \label{Lm:reducing}
Let $G$ be an ISK4-free graph. Let $P$ be a flat path of length at least $2$ in $G$ and $G'$ be the graph obtained from $G$ by reducing $P$. Then $G'$ is ISK4-free.
\end{lemma}

\begin{proof}
Let $e$ be the edge of $G'$ that results from the reduction of $P$. Suppose that $G'$ contains an ISK4 $H$. Then $H$ must contain $e$, for otherwise $H$ is an ISK4 in $G$. Then replacing $e$ by $P$ in $H$ yields an ISK4 in $G$, contradiction. 
\end{proof}

It is shown in \cite{LMT12} that clique cutsets and proper $2$-cutsets are useful for proving Conjecture \ref{conj:1} in the inductive sense. If we can find such a cutset in $G$, then we immediately have a bound for the chromatic number of $G$, since $\chi(G)\leq \max\{\chi(G_1),\chi(G_2)\}$, where $G_1$ and $G_2$ are two blocks of decomposition of $G$ with respect to that cutset (see the proof of Theorem~ 1.4 in \cite{LMT12}). Therefore, we only have to prove Conjecture \ref{conj:1} for the class of $\{$ISK4, $K_{3,3}$, prism, $K_{2,2,2}\}$-free graphs and prove Conjecture \ref{conj:2} for the class of $\{$ISK4, $K_{3,3}$, triangle$\}$-free graphs since the existence of $K_{3,3}$, prism or $K_{2,2,2}$ implies a good cutset by Lemmas \ref{Lm:K33} and \ref{Lm:prism,K222}. 

We say that $S$ \textit{dominates} $C$ if $N_C(S)=C$. The \textit{distance} between two vertices $x$, $y$ in $V(G)$ is the length of a shortest path from $x$ to $y$ in $G$. Let $u\in V(G)$ and $i$ be an integer, denote by $N_i(u)$ the set of vertices of $G$ that are of distance exactly $i$ from $u$. Note that there are no edges between $N_i(u)$ and $N_j(u)$ for every $i,j$ such that $|i-j|\geq 2$. 

\begin{lemma} \label{Lm:upstair-path}
Let $G$ be a graph, $u\in V(G)$ and $i$ be an integer $\geq 1$. Let $x,y$ be two distinct vertices in $N_i(u)$. Then, there exists a path $P$ in $G$ from $x$ to $y$ such that $V(P)\subseteq \{u\}\cup N_1(u)\cup \ldots\cup N_{i}(u)$ and $|V(P)\cap N_j(u)|\leq 2$ for every $j\in\{1,\ldots,i\}$.
\end{lemma}

\begin{proof}
We prove this by induction on $i$. If $i=1$, we have $x,y\in N_1(u)$. If $xy\in E(G)$, we choose $P=xy$, otherwise, choose $P=xuy$. Suppose that the lemma is true until $i=k$, we prove that it is also true for $i=k+1$. If $xy\in E(G)$, we choose $P=xy$. Otherwise, let $x',y'$ be the vertices in $N_k(u)$ such that $x'x,y'y\in E(G)$. If $x'=y'$, we choose $P=xx'y$, otherwise choose $P=P'\cup \{x,y\}$, where $P'$ is the path with two ends $x'$ and $y'$ generated by applying induction hypothesis. 
\end{proof}

Such a path $P$ in Lemma \ref{Lm:upstair-path} is called the \textit{upstairs path} of $\{x,y\}$. For three distinct vertices $x,y,z\in V(G)$, a graph $H$ is a \textit{confluence} of $\{x,y,z\}$ if it is one of the two following types:
\begin{itemize}
	\item Type $1$: 
	\begin{itemize}
		\item $V(H)=V(P_x)\cup V(P_y)\cup V(P_z)$.
		\item $P_x$, $P_y$, $P_z$ are three paths having a common end $u$ and $P_x\setminus u$, $P_y\setminus u$, $P_z\setminus u$ are pairwise disjoint. The other ends of $P_x$, $P_y$, $P_z$ are $x$, $y$, $z$, respectively. 
		\item These are the only edges in $H$.
	\end{itemize}
	\item Type $2$:
	\begin{itemize}
		\item $V(H)=V(P_x)\cup V(P_y)\cup V(P_z)$.
		\item $P_x$ is a path with two ends $x$ and $x'$.
		\item $P_y$ is a path with two ends $y$ and $y'$.
		\item $P_z$ is a path with two ends $z$ and $z'$.
		\item $P_x$, $P_y$, $P_z$ are pairwise disjoint.
		\item $x'y'z'$ is a triangle.
		\item These are the only edges in $H$.
	\end{itemize}
\end{itemize}

If $H$ is a confluence of Type $1$, the vertex $u$ is called the \textit{center} of $H$ and if $H$ is a confluence of Type $2$, the triangle $x'y'z'$ is called the \textit{center triangle} of $H$. Note that the length of $P_x$ can be $0$ when $x=u$ (for Type $1$) or $x=x'$ (for Type $2$).

\begin{lemma} \label{Lm:confluence}
Let $G$ be a graph, $u\in V(G)$ and $i$ be an integer $\geq 1$. Let $x,y,z$ be three distinct vertices in $N_i(u)$. Then, there exists a set $S\subseteq \{u\}\cup N_1(u)\cup \ldots\cup N_{i-1}(u)$ such that $G|(S\cup\{x,y,z\})$ is a confluence of $\{x,y,z\}$.
\end{lemma}

\begin{proof}
Let $G'$ be the subgraph of $G$ induced by $\{u\}\cup N_1(u)\cup \ldots\cup N_{i-1}(u)$. It is clear that $G'$ is connected. Let $P$ be a path in $G'$ from $x$ to $y$ and $Q$ be a path in $G'$ from $z$ to $P$ (one end of $Q$ is in $P$). We choose $P$ and $Q$ subject to minimize $|V(P\cup Q)|$. It is easy to see that $G|V(P\cup Q)$ is a confluence of $\{x,y,z\}$.
\end{proof}

The notions of upstairs path and confluence are very useful to find induced structures in our graph since they establish a way to connect two or three vertices of the same layer through only the upper layers. 

\begin{lemma} \label{Lm:chromatic-layer}
Let $G$ be a graph and $u\in V(G)$. Then: $$\chi(G)\leq \max_{i\textrm{ odd}}\chi(G|N_i(u))+\max_{j \textrm{ even}}\chi(G|N_j(u)).$$
\end{lemma}

\begin{proof}
It is clear that in $G$, there are no edges between $N_i(u)$ and $N_j(u)$ if $i\neq j$ and $i,j$ are of the same parity. Therefore, we can color all the odd layers with $\max_{i\textrm{ odd}}\chi(G|N_i(u))$ colors and all the even layers with $\max_{j \textrm{ even}}\chi(G|N_j(u))$ other colors. The lemma follows.
\end{proof}

\section{Proof of Theorem \ref{Thm:1}} \label{S:3}

The next lemma shows that if there is a set $S$ that dominates some hole $C$, then there must exist some vertices in $S$ which have very few (one or two) neighbors in $C$.

\begin{lemma} \label{Lm:attach-hole}
Let $G$ be an $\{$ISK4, triangle, $K_{3,3}\}$-free graph and $C$ be a hole in $G$. Let $S\subseteq V(G)\setminus C$ be such that every vertex in $S$ has at least a neighbor in $C$ and $S$ dominates $C$. Then one of the following cases holds:
\begin{enumerate}
	\item \label{Lm:attach-hole:1} There exist four distinct vertices $u_1$, $u_2$, $u_3$, $u_4$ in $S$ and four distinct vertices $v_1$, $v_2$, $v_3$, $v_4$ in $C$ such that for $i\in\{1,2,3,4\}$, $N_C(u_i)=\{v_i\}$.
	\item \label{Lm:attach-hole:2} There exist three distinct vertices $u_1$, $u_2$, $u_3$ in $S$ and three distinct vertices $v_1$, $v_2$, $v_3$ in $C$ such that for $i\in\{1,2,3\}$, $N_C(u_i)=\{v_i\}$ and $v_1$, $v_2$, $v_3$ are pairwise non-adjacent.
	\item \label{Lm:attach-hole:3} There exist three distinct vertices $u_1$, $u_2$, $u_3$ in $S$ and four distinct vertices $v_1$, $v_2$, $v_3$, $v_3'$ in $C$ such that $N_C(u_1)=v_1$, $N_C(u_2)=v_2$, $N_C(u_3)=\{v_3,v_3'\}$ and $v_1$, $v_3$, $v_2$, $v_3'$ appear in this order along $C$.
\end{enumerate}
\end{lemma}

\begin{proof}
We prove Lemma \ref{Lm:attach-hole} by induction on the length of hole $C$ for every $\{$ISK4, triangle, $K_{3,3}\}$-free graph. First, suppose that the length of $C$ is $4$ and $C=c_0c_1c_2c_3$. Since $G$ is triangle-free, a vertex in $S$ can only have one or two neighbors in $C$. We consider two cases:
\begin{itemize}
	\item If some vertex $u\in S$ has two neighbors in $C$, w.l.o.g, suppose $N_C(u)=\{c_0,c_2\}$. Since $S$ dominates $C$, there exists some vertices $v$, $w\in S$ such that $vc_1, wc_3\in E$. If $v=w$ then $\{u,v,w\}\cup C$ induces $K_{3,3}$ (if $uv\in E$) or an ISK4 (if $uv\notin E$), contradiction. Therefore, $v\neq w$ and $u,v,w$ are three vertices satisfying output \ref{Lm:attach-hole:3} of the lemma.
	\item If every vertex in $S$ has exactly one neighbor in $C$, output \ref{Lm:attach-hole:1} of the lemma holds.
\end{itemize} 
	Now, we may assume that $|C|\geq 5$ and the lemma is true for every hole of length at most $|C|-1$. A vertex $u\in S$ is a \textit{bivertex} if $N_C(u)=\{u',u''\}$ and the two paths $P_1$, $P_2$ from $u'$ to $u''$ in $C$ are of lengths at least $3$. Suppose that $S$ contains such a bivertex $u$. Let $C_1=P_1\cup \{u\}$, $C_2=P_2\cup \{u\}$, note that $|C_1|,|C_2|<|C|$. Consider the graph $G'$ obtained from $G$ as follows: $V(G')=V(G)\cup \{a,b,c\}$, $E(G')=E(G)\cup \{au,bu',cu''\}$. It is clear that $G'$ is $\{$ISK4, triangle, $K_{3,3}\}$-free. Let $S_1=\{v\in S\setminus u|N_{C_1}(v)\neq \emptyset\}\cup \{a,b,c\}$ and $S_2=\{v\in S\setminus u|N_{C_2}(v)\neq \emptyset\}\cup \{a,b,c\}$. By applying the induction hypothesis on $S_1$ and $C_1$, we obtain that there is some vertex $x\in S$ such that $x$ has exactly one neighbor in $P_1$ which is in $P_1^*$ ($x$ can be adjacent to $u$). We claim that $x$ has exactly one neighbor in $C$. Indeed, if $x$ has exactly one neighbor $x'$ in $P_2^*$ then $C\cup \{x,u\}$ induces an ISK4 (if $xu\notin E(G)$) or $C_1\cup \{x\}\cup Q$ induces an ISK4 (if $xu\in E(G)$), where $Q$ is the shorter path in one of the two paths in $C$: $x'P_2u'$ and $x'P_2u''$, contradiction. If $x$ has at least two neighbors in $P_2^*$, let $x'$, $x''$ be the neighbors of $x$ closest to $u'$, $u''$ on $P_2^*$, respectively. Then $C_1\cup \{x\}\cup x'P_2u'\cup x''P_2u''$ induces an ISK4 (if $xu\notin E(G)$) or $C_1\cup \{x\}\cup x'P_2u'$ induces an ISK4 (if $xu\in E(G)$), contradiction. So, $x$ has no neighbor in $P_2^*$ and has exactly one neighbor in $C$ as claimed. Similarly, by applying the induction hypothesis on $S_2$ and $C_2$, we know that there is some vertex $y\in S$ such that $y$ has exactly one neighbor in $P_2^*$ and this is also its only neighbor in $C$. Now, $\{x,y,u\}$ satisfies output \ref{Lm:attach-hole:3} of the lemma. Hence, we may assume that $S$ contains no bivertex.
	
Suppose that there is some vertex $u$ in $S$ which has at least four neighbors in $C$. Let $N_C(u)={u_0,\ldots,u_k}$ where $u_0,\ldots,u_k$ ($k\geq 3$) appear in that order along $C$. Let $P_u(i,i+3)$ be the path of $C$ from $u_i$ to $u_{i+3}$ which contains $u_{i+1}$ and $u_{i+2}$ and define $\amp(u,C)=\max_{i=0}^k |P_u(i,i+3)|$ (the index is taken in modulo $k+1$). Note that this notion is defined only for a vertex with at least four neighbors in $C$. Let $v\in S$ be such that $\amp(v,C)$ is maximum. W.l.o.g suppose that $P_v(0,3)$ is the longest path among all paths of the form $P_v(i,i+3)$. Let $P_0$, $P_1$, $P_2$ be the subpaths of $P_v(0,3)$ from $v_0$ to $v_1$, $v_1$ to $v_2$, $v_2$ to $v_3$, respectively. Let $C_0=\{v\}\cup P_0$, $C_1=\{v\}\cup P_1$ and $C_2=\{v\}\cup P_2$. Consider the graph $G'$ obtained from $G$ as follows: $V(G')=V(G)\cup \{a,b,c\}$, $E(G')=E(G)\cup \{av,bv_0,cv_1\}$. It is clear that $G'$ is $\{$ISK4, triangle, $K_{3,3}\}$-free. Let $S_0=\{u\in S\setminus v|N_{C_0}(u)\neq \emptyset\}\cup \{a,b,c\}$. By applying the induction hypothesis on $S_0$ and $C_0$, we obtain that there is some vertex $x\in S$ such that $x$ has exactly one neighbor $x_0$ in $P_0$ which is in $P_0^*$ ($x$ can be adjacent to $v$). We claim that $x$ has exactly one neighbor in $C$. Suppose that $x$ has some neighbor in $P_1$. Let $x_1$, $x_2$ be the neighbors of $x$ in $P_1$ which is closest to $v_1$ and $v_2$, respectively ($x_1$ and $x_2$ could be equal). Then we have $\{x,v\}\cup P_0\cup v_1P_1x_1\cup v_2P_1x_2$ induces an ISK4 (if $xv\notin E(G)$) or $\{x,v\}\cup P_0\cup v_1P_1x_1$ induces an ISK4 (if $xv\in E(G)$), contradiction. Therefore, $x$ has no neighbor in $P_1$. Suppose that $x$ has some neighbor in $P_2$, let $x_1$ be the neighbor of $x$ in $P_2$ which is closest to $v_2$. Let $Q$ be the path from $x_0$ to $x_1$ in $C$ which contains $v_1$. We have $\{x,v\}\cup Q\cup v_0P_0x_0$ induces an ISK4 (if $xv\notin E(G)$) or $\{x,v\}\cup Q$ induces an ISK4 (if $xv\in E(G)$), contradiction. Hence, $x$ has no neighbor in $P_2$. Now if $x$ has at least four neighbors in $C$, $\amp(x,C)>\amp(v,C)$, contradiction to the choice of $v$. Hence, $x$ can have at most one neighbor in the path from $v_0$ to $v_3$ in $C$ which does not contain $v_1$. Suppose $x$ has one neighbor $x'$ in that path. By the assumption that we have no bivertex, $x'v_0, v_0x_0\in E(G)$. Let $Q$ be the path from $v_{-1}$ to $x'$ in $C$ which does not contain $v_0$. We have $\{x,x',v_0,x_0,v\}\cup Q\cup v_1P_0x_0$ induces an ISK4 (if $xv\notin E(G)$) or $\{x,x',v_0,x_0,v\}\cup Q$ induces an ISK4 (if $xv\in E(G)$), contradiction. Hence, $x_0$ is the only neighbor of $x$ in $C$, as claimed. Similarly, we can prove that there exist two vertices $y,z\in S$ such that they have exactly one neighbor in $C$ which are in $P_1^*$, $P_2^*$, respectively. Note that the proof for $y$ is not formally symmetric to the one for $x$ and $z$, but the proof is actually the same. In particular, a vertex $y$ with a unique neighbor in $P_1^*$, no neighbor in $P_0$, $P_2$ and at least four neighbors in $C$ also yields a contradiction to the maximality of $\amp(v,C)$.  Therefore, $\{x,y,z\}$ satisfies output \ref{Lm:attach-hole:2} of the lemma. Now, we can assume that no vertex in $S$ has at least four neighbors in $C$. 

Hence, every vertex in $S$ either has exactly one neighbor in $C$ or exactly two neighbors in $C$ and is not a bivertex. Suppose there is some vertex $u$ that has two neighbors $u'$, $u''$ on $C$ and let $x\in C$ be such that $xu',xu''\in E$. Let $v\in S$ be a vertex adjacent to $x$. If $v$ has another neighbor $x'$ in $C$ then $x'$ must be adjacent to $u'$ or $u''$, since $v$ is not a bivertex. So, we have $\{u,v,x',u',x,u''\}$ induces an ISK4 (if $uv\in E(G)$) or $\{u,v\}\cup C$ induces an ISK4 (if $uv\notin E(G)$), contradiction. So, $v$ has only one neighbor $x$ in $C$. Hence, if we have at least one vertex which has two neighbors on $C$, the output \ref{Lm:attach-hole:3} holds. If every vertex has exactly one neighbor in $C$, the output \ref{Lm:attach-hole:1} holds, which completes the proof.     	
\end{proof}

\begin{lemma} \label{Lm:hole}
Let $G$ be an $\{$ISK4, triangle, $K_{3,3}\}$-free graph and $u\in V(G)$. For every $i\geq 1$, $G|N_i(u)$ does not contain any hole.
\end{lemma}
\begin{proof}
Suppose for some $i$, $G|N_i(u)$ contains a hole $C$. For every vertex $v\in C$, there exists a vertex $v'\in N_{i-1}(u)$ such that $vv'\in E$. Hence there exists a subset $S\subseteq N_{i-1}(u)$ such that $S$ dominates $C$. Let us apply Lemma \ref{Lm:attach-hole} for $S$ and $C$:
\begin{itemize}
	\item If output \ref{Lm:attach-hole:1} or \ref{Lm:attach-hole:2} of Lemma \ref{Lm:attach-hole} holds, then there exist three distinct vertices $u_1$, $u_2$, $u_3$ in $S$ and three distinct vertices $v_1$, $v_2$, $v_3$ in $C$ such that for $i\in\{1,2,3\}$, $N_C(u_i)=\{v_i\}$. By Lemma \ref{Lm:confluence}, since $G$ is triangle-free, there exists a confluence $F$ of $\{u_1,u_2,u_3\}$ of Type $1$, so $F\cup C$ induces an ISK4, contradiction.
	\item If output \ref{Lm:attach-hole:3} of Lemma \ref{Lm:attach-hole} holds, then there exist two distinct vertices $u_1$, $u_2$ in $S$ and three distinct vertices $v_1$, $v_2$, $v_2'$ in $C$ such that $N_C(u_1)=v_1$, $N_C(u_2)=\{v_2,v_2'\}$. By Lemma \ref{Lm:upstair-path}, there exists an upstairs path $P$ of $\{u_1,u_2\}$, so $P\cup C$ induces an ISK4, contradiction.
\end{itemize} 
\end{proof}

\begin{proof}[Proof of Theorem \ref{Thm:1}]
We prove the theorem by induction on the number of vertices of $G$. Suppose that $G$ has a clique cutset $K$. So $G\setminus K$ can be partitioned into two sets $X$, $Y$ such that there is no edge between them. By induction hypothesis $\chi(G|(X\cup K))$ and $\chi(G|(Y\cup K))\leq 4$, therefore $\chi(G)\leq \max\{\chi(G|(X\cup K)),\chi(G|(Y\cup K))\}\leq 4$. Hence we may assume that $G$ has no clique cutset. If $G$ contains a $K_{3,3}$, then by Lemma \ref{Lm:K33}, $G$ is a thick complete bipartite graph and $\chi(G)\leq 2$. So we may assume that $G$ contains no $K_{3,3}$. By Lemma \ref{Lm:hole}, for every $u\in V(G)$, for every $i\geq 1$, $G|N_i(u)$ is a forest, hence $\chi(G|N_i(u))\leq 2$. By Lemma \ref{Lm:chromatic-layer}, $\chi(G)\leq 4$, which completes the proof.
\end{proof}

\section{Proof of Theorem \ref{Thm:2}} \label{S:4}

A \textit{boat} is a graph consisting of a hole $C$ and a vertex $v$ that has exactly four consecutive neighbors in $C$ ($N_C(v)$ induces a $C_4$ if $|C|=4$ or a $P_4$ if $|C|\geq 5$). A \textit{4-wheel} is a particular boat whose hole is of length $4$. Let ${\cal C}_1$ be the class of $\{$ISK4, $K_{3,3}$, prism, boat$\}$-free graphs, ${\cal C}_2$ be the class of $\{$ISK4, $K_{3,3}$, prism, $4$-wheel$\}$-free graphs and ${\cal C}_3$ be the class of $\{$ISK4, $K_{3,3}$, prism, $K_{2,2,2}\}$-free graphs. Remark that ${\cal C}_1\subsetneq {\cal C}_2\subsetneq {\cal C}_3\subsetneq$ ISK4-free graphs.

\begin{lemma} \label{Lm:boat-free}
Let $G$ be a graph in ${\cal C}_1$. Then $\chi(G)\leq 6$.
\end{lemma}

\begin{proof}
We prove first the following claim.
\begin{claim} \label{Cl:no triangle and C_4}
Let $u\in V(G)$ and $i\geq 1$. Then $G|N_i(u)$ contains no triangle and no $C_4$.
\end{claim}
\begin{proof}
Suppose $G|N_i(u)$ contains a triangle $abc$. No vertex is complete to $abc$ since $G$ is $K_4$-free. Suppose that there is some vertex $x\in N_{i-1}(u)$ which has exactly two neighbors in the triangle, w.l.o.g. assume that they are $a$ and $b$. Let $y$ be some vertex in $N_{i-1}(u)$ adjacent to $c$ and $P$ be an upstairs path of $\{x,y\}$. If $y$ has exactly one neighbor in $abc$ (which is $c$), then $P\cup \{a,b,c\}$ induces an ISK4, contradiction. Hence $y$ must have another neighbor in $C$, say $a$ up to symmetry. In this case, $P\cup \{a,b,c\}$ induces a boat, contradiction. Then every vertex in $N_{i-1}(u)$ has exactly one neighbor in $abc$. Suppose there are three vertices $x,y,z\in N_{i-1}(u)$ such that $N_{abc}(x)=\{a\}$, $N_{abc}(y)=\{b\}$ and $N_{abc}(z)=\{c\}$. By Lemma \ref{Lm:confluence}, there exists a confluence $S$ of $\{x,y,z\}$. If $S$ is of Type $1$, then $S\cup \{a,b,c\}$ induces an ISK4, contradiction. If $S$ is of Type $2$, then $S\cup \{a,b,c\}$ induces a prism, contradiction. Hence, $G|N_i(u)$ contains no triangle.

Suppose $N_i(u)$ contains a $C_4$, namely $abcd$. Every vertex can only have zero, one or two neighbors in $abcd$ since a $4$-wheel is a boat. Suppose there is some vertex $x\in N_{i-1}(u)$ which has exactly two non-adjacent neighbors in $\{a,b,c,d\}$, say $N_{abcd}(x)=\{a,c\}$. Let $y$ be some vertex in $N_{i-1}(u)$ adjacent to $d$ and $P$ be an upstairs path of $\{x,y\}$. If $yb\in E$, then $\{x,y,a,b,c,d\}$ induces an ISK4 (if $xy\notin E$) or a $K_{3,3}$ (if $xy\in E$), contradiction. If $ya\in E$, $P\cup \{a,c,d\}$ induces an ISK4, contradiction. Then $yc\notin E$ also by symmetry, and $N_{abcd}(y)=\{d\}$. In this case $P\cup \{a,b,c,d\}$ induces an ISK4, contradiction. Therefore, there is no vertex in $N_{i-1}(u)$ has two non-adjacent neighbors in $abcd$. Now, suppose that there is some vertex $x\in N_{i-1}(u)$ which has exactly two consecutive neighbors $\{a,b\}$ in $abcd$. Let $y$ be some vertex in $N_{i-1}(u)$ adjacent to $d$ and $P$ be an upstairs path of $\{x,y\}$. If $y$ is adjacent to $c$, then $P\cup \{a,b,c,d\}$ induces a prism, contradiction. If $N_{abcd}(y)=\{d\}$, then $P\cup \{a,b,c,d\}$ induces an ISK4, contradiction. Hence $N_{abcd}(y)=\{a,d\}$. Let $z$ be some vertex in $N_{i-1}(u)$ adjacent to $c$, $P_{xz}$ be an upstairs path of $\{x,z\}$ and $P_{yz}$ be an upstairs path of $\{y,z\}$. If $zb\in E$, $P_{yz}\cup \{a,b,c,d\}$ induces a prism, contradiction. If $zd\in E$, $P_{xz}\cup \{a,b,c,d\}$ induces a prism, contradiction. Hence $N_{abcd}(z)=\{c\}$. In this case, $P_{xz}\cup \{a,b,c,d\}$ induces an ISK4, contradiction. Therefore, there is no vertex in $N_{i-1}(u)$ having two neighbors in $abcd$. So, there are three vertices $x,y,z\in N_{i-1}(u)$ such that $N_{abcd}(x)=\{a\}$, $N_{abcd}(y)=\{b\}$, $N_{abcd}(z)=\{c\}$. By Lemma \ref{Lm:confluence}, there exists a confluence $S$ of $\{x,y,z\}$. If $S$ is of Type $1$, $S\cup \{a,b,c,d\}$ induces an ISK4, contradiction. If $S$ is of Type $2$, $S\cup \{a,b,c\}$ induces an ISK4, contradiction. Therfore, $G|N_i(u)$ contains no $C_4$.      
\end{proof}
By Claim \ref{Cl:no triangle and C_4}, the girth of $N_i(u)$ is at least $5$ for $i\geq 1$. By Theorem \ref{Thm:girth5}, $\chi(G|N_i(u))\leq 3$. By Lemma \ref{Lm:chromatic-layer}, $\chi(G)\leq 6$, which completes the proof.
\end{proof}

\begin{lemma} \label{Lm:4-wheel-free}
Let $G$ be a graph in ${\cal C}_2$. Then $\chi(G)\leq 12$.
\end{lemma}

\begin{proof}
We first prove that: for any $u\in V(G)$ and $i\geq 1$, $G|N_i(u)$ contains no boat. We may assume that $i\geq 2$, since $G|N_1(u)$ is triangle-free, the conclusion holds for $i=1$. Suppose for contradiction that $G|N_i(u)$ contains a boat consisting of a hole $C$ and a vertex $x$ that has four neighbors $a$, $b$, $c$, $d$ in this order on $C$. Since $G$ contains no $4$-wheel, we can assume that $|C|\geq 5$ and $\{a,b,c,d\}$ induces a $P_4$. Let $P$ be the path from $a$ to $d$ in $C$ which does not go through $b$.

\begin{claim} \label{Cl:bc}
No vertex in $N_{i-1}(u)$ is adjacent to both $b$ and $c$.
\end{claim}
\begin{proof}
Suppose there is a vertex $y\in N_{i-1}(u)$ adjacent to both $b$ and $c$. Since $\{x,y,b,c\}$ does not induce $K_4$, $xy\notin E$. If $ya\in E$, $\{a,b,c,x,y\}$ induces a $4$-wheel, contradiction. Hence, $ya\notin E$. We also have $yd\notin E$ by symmetry. We claim that $N_{C}(y)=\{b,c\}$. Suppose that $y$ has some neighbor in $P^*$. If $y$ has exactly one neighbor in $P^*$, then $\{y\}\cup C$ induces an ISK4, contradiction. If $y$ has exactly two consecutive neighbor in $P^*$, then $C\cup \{x,y\}\setminus \{c\}$ induces a prism, contradiction. If $y$ has at least three neighbors in $P^*$, or two neighbors in $P^*$ that are not consecutive, then let $z$ be the one closest to $a$ and $t$ be the one closest to $d$. Then $\{x,y,b\}\cup zPa\cup tPd$ induces an ISK4, contradiction. So $N_{C}(y)=\{b,c\}$. Let $z$ be a vertex in $N_{i-1}(u)$ which has a neighbor in $P^*$ and $P_{yz}$ be an upstairs path of $\{y,z\}$. If $z$ has exactly one neighbor in $C$, then $P_{yz}\cup C$ induces an ISK4, contradiction. If $z$ has exactly two consecutive neighbors in $C$, then $P_{yz}\cup C$ induces a prism, contradiction. If $z$ has at least three neighbors in $C$ or two neighbors in $C$ which are not consecutive, let $t,w$ be the ones closest to $b,c$ in $C$, respectively. Let $Q$ be the path form $t$ to $w$ in $C$ which contains $b$. We have that $P_{yz}\cup Q$ induces an ISK4, contradiction.  
\end{proof}
By Claim \ref{Cl:bc}, let $y,z$ be two distinct vertices in $N_{i-1}(u)$ such that $yb,zc\in E$ and $P_{yz}$ be an upstairs path of $\{y,z\}$. 

\begin{claim}
$xy,xz\in E$.
\end{claim}
\begin{proof}
Suppose $xy\notin E$. Then $xz\notin E$, otherwise $P_{yz}\cup \{x,b,c\}$ induces an ISK4. Let $t\in N_{i-1}(u)$ such that $tx\in E$, let $P_{ty}$ and $P_{tz}$ be an upstairs paths of $\{t,y\}$ and $\{t,z\}$, respectively. If $tb\in E$, then $P_{tz}\cup \{x,b,c\}$ induces an ISK4, contradiction. If $tc\in E$, then $P_{ty}\cup \{x,b,c\}$ induces an ISK4, contradiction. So $N_{xbc}(t)=\{x\}$. By Lemma \ref{Lm:confluence}, let $S$ be a confluence of $\{y,z,t\}$. If $S$ is of Type $1$, $S\cup \{x,b,c\}$ induces an ISK4, contradiction. If $S$ is of Type $2$, $S\cup \{x,b,c\}$ induces a prism, contradiction. Then $xy\in E$. Symmetrically, $xz\in E$.
\end{proof}

\begin{claim} \label{Cl:only}
$N_C(y)=\{b\}$ and $N_C(z)=\{c\}$.
\end{claim}
\begin{proof}
We prove only $N_C(y)=\{b\}$, the other conclusion is proved similarly. First, $ya,yc\notin E$, otherwise $\{y,x,a,b\}$ or $\{y,x,a,c\}$ induces a $K_4$. We also have $yd\notin E$, otherwise $\{x,y,b,c,d\}$ induces a $4$-wheel. If $y$ has some neighbor in $P^*$, let $t$ be the one closest to $a$. In this case, $tPa\cup \{x,y,b\}$ induces an ISK4, contradiction. Hence $N_C(y)=\{b\}$.
\end{proof}

Let $t$ be a vertex in $N_{i-1}(u)$ such that $ta\in E$ and $P_{yt}$ be an upstairs path of $\{y,t\}$. By Claim \ref{Cl:only}, $tb,tc\notin E$. We have $tx\in E$, otherwise $P_{yt}\cup \{x,a,b\}$ induces an ISK4.  Suppose that $N_C(t)=\{a\}$. There exists a confluence $S$ of $\{t,y,z\}$ by Lemma \ref{Lm:confluence}. If $S$ is of Type $1$, $S\cup C$ induces an ISK4, contradiction. If $S$ is of Type $2$, $S\cup \{a,b,c\}$ induces an ISK4, contradiction. Hence, $t$ must have some neighbor in $P\setminus \{a\}$, let $w$ be the one closest to $d$ along $P$ and $P_w$ be the path from $a$ to $w$ in $C$ which contains $b$. 

\begin{claim}
$t$ has some neighbor in $P_{yz}$.
\end{claim}
\begin{proof}
Suppose that $t$ has no neighbor in $P_{yz}$. Because $G|(u\cup N_1(u)\cup\ldots\cup N_{i-2}(u))$ is connected, there exists a path $Q$ from $t$ to some $t'$ such that $Q\setminus \{t\}\subseteq u\cup N_1(u)\cup\ldots\cup N_{i-2}(u)$ and $t'$ is the only vertex in $Q$ which has some neighbor in $P_{yz}$. If $t'$ has exactly one neighbor in $P_{yz}$, then $P_w\cup Q\cup P_{yz}$ induces an ISK4, contradiction. If $t'$ has exactly two consecutive neighbors in $P_{yz}$, then $Q\cup P_{yz}\cup \{a,b,c\}$ induces an ISK4. If $t'$ has at least three neighbors in $P_{yz}$ or two neighbors in $P_{yz}$ which are not consecutive, let $y'$, $z'$ be the one closest to $y$, $z$, respectively, then $Q\cup P_w\cup y'P_{yz}y\cup z'P_{yz}z$ induces an ISK4, contradiction. Then $t$ must have some neighbor in $P_{yz}$.
\end{proof}

Let $y',z'\in P_{yz}$ such that $y'y,z'z\in E$. Since $t\in N_{i-1}(u)$, $N_{P_{yz}}(t)\subseteq \{y,z,y',z'\}$. If $t$ has exactly one neighbor in $P_{yz}$, then $\{t\}\cup P_{yz}\cup P_w$ induces an ISK4, contradiction. If $t$ has exactly two neighbors in $P_{yz}$, then $\{t,a,b,c\}\cup P_{yz}$ induces an ISK4, contradiction. If $t$ has exactly three neighbors in $P_{yz}$, then $\{t,b,c\}\cup P_{yz}$ induces an ISK4, contradiction. Hence, $t$ has four neighbors in $P_{yz}$ or $N_{P_{yz}}(t)= \{y,z,y',z'\}$. In particular, $ty\in E$ and $\{x,t,y,a,b\}$ induces a $4$-wheel, contradiction. Hence, $G|N_i(u)$ is boat-free.

Now, for every $i\geq 1$, $G|N_i(u)\in {\cal C}_1$. By Lemma \ref{Lm:boat-free}, $\chi(G|N_i(u))\leq 6$ . By Lemma \ref{Lm:chromatic-layer}, $\chi(G)\leq 12$, completing the proof.
\end{proof}

\begin{lemma} \label{Lm:K222-free}
Let $G$ be a graph in ${\cal C}_3$. Then $\chi(G)\leq 24$.
\end{lemma}

\begin{proof}
Let $u\in V(G)$ and $i\geq 1$. We claim that $G|N_i(u)$ contains no $4$-wheel. Suppose that $G|N_i(u)$ contains a $4$-wheel consisting of a hole $abcd$ and a vertex $x$ complete to $abcd$. By similar argument as in the proof of Lemma \ref{Lm:boat-free} (the proof of $C_4$-free), the hole $abcd$ cannot be dominated by only the vertices in $N_{i-1}(u)$ which has one or two neighbors in $abcd$. Hence, there exists some vertex $y\in N_{i-1}(u)$ complete to $abcd$. It is clear that $xy\notin E$, otherwise $\{x,y,a,b\}$ induces a $K_4$. Now, $\{x,y,a,b,c,d\}$ induces a $K_{2,2,2}$, contradiction. So, $G|N_i(u)$ contains no $4$-wheel. By Lemma \ref{Lm:4-wheel-free}, $\chi(G|N_i(u))\leq 12$. By Lemma \ref{Lm:chromatic-layer}, we have $\chi(G)\leq 24$, which proves the lemma.   
\end{proof}

Before the main proof, we have several lemmas proving the bound of chromatic number of some basic graphs.

\begin{lemma}\label{Lm:line-graph}
Let $G$ be the line graph of a graph $H$ with maximum degree three. Then $\chi(G)\leq 4$.
\end{lemma}
\begin{proof}
To prove that $G$ is $4$-colorable, we only need to prove that $H$ is $4$-edge-colorable. But since the maximum degree of $H$ is three, this is a direct consequence of Vizing's theorem (see \cite{BM76}). 
\end{proof}

\begin{lemma}\label{Lm:rich-square}
Let $G$ be a rich square. Then $\chi(G)\leq 4$.
\end{lemma}
\begin{proof}
By the definition of a rich square, there is a square $S=\{u_1,u_2,u_3,u_4\}$ in $G$ such that every component of $G\setminus S$ is a link of $S$. We show a $4$-coloring of $G$ as follows. Assign color $1$ to $\{u_1,u_3\}$ and color $2$ to $\{u_2,u_4\}$. Let $P$ be a component of $G\setminus S$ with two ends $p$, $p'$. If $p=p'$, give it color $3$. If $p\neq p'$, give $p$, $p'$ color $3$, $4$, respectively and assign color $1$ and $2$ alternately to the internal vertices of $P$.  
\end{proof}

\begin{proof}[Proof of Theorem \ref{Thm:2}]
We prove the theorem by induction on the number of vertices of $G$. Suppose that $G$ has a clique cutset $K$. So $G\setminus K$ can be partitioned into two sets $X$, $Y$ such that there are no edges between them. By the induction hypothesis, $\chi(G|(X\cup K))$ and $\chi(G|(Y\cup K))\leq 24$, therefore $\chi(G)\leq \max\{\chi(G|(X\cup K)),\chi(G|(Y\cup K))\}\leq 24$. Hence we may assume that $G$ has no clique cutset. If $G$ contains a $K_{3,3}$, then by Lemma \ref{Lm:K33}, $G$ is a thick complete bipartite graph or complete tripartite graph and $\chi(G)\leq 3$. So we may assume that $G$ contains no $K_{3,3}$. 

Suppose that $G$ has a proper $2$-cutset $\{a,b\}$. So $G\setminus \{a,b\}$ can be partitioned into two sets $X$, $Y$ such that there is no edge between them. Since $G$ has no clique cutset, it is $2$-connected, so there exists a path $P_Y$ with ends $a$ and $b$ and with interior in $Y$. Let $G_X'$ be the subgraph of $G$ induced by $X\cup P_Y$. Note that $P_Y$ is a flat path in $G_X'$. Let $G_X''$ be obtained from $G_X'$ by reducing $P_Y$. Define a graph $G_Y''$ similarly. Since $G_X'$ is an induced subgraph of $G$, it contains no ISK4. So, by Lemma \ref{Lm:reducing}, $G_X''$ contains no ISK4. The same hold for $G_Y''$. By induction hypothesis, $G_X''$ and $G_Y''$ admit a $24$-coloring. Since $a$ and $b$ have different colors in both coloring, we can combine them so that they coincide on $\{a,b\}$ and obtain a $24$-coloring of $G$. Now, we may assume that $G$ has no proper $2$-cutset. If $G$ contains a $K_{2,2,2}$ (rich square) or a prism, then by Lemma \ref{Lm:prism,K222}, $G$ is the line graph of a graph with maximum degree $3$, or a rich square. By Lemmas \ref{Lm:line-graph} and \ref{Lm:rich-square}, $\chi(G)\leq 4<24$. Therefore, we may assume that $G$ contains neither prism nor $K_{2,2,2}$. So $G\in {\cal C}_3$ and $\chi(G)\leq 24$ by Lemma \ref{Lm:K222-free}.
\end{proof}

\section{Conclusion}

Not only the bound we found in Theorem \ref{Thm:1} is very close to the one stated in Conjecture \ref{conj:2}, but the simple structure of each layer is also interesting. We believe that it is very promising to settle Conjecture \ref{conj:2} by this way of looking at our class. For Theorem \ref{Thm:2}, we are convinced that the bound $24$ we found could be slightly improved by this method if we look at each layer more carefully and exclude more structures, but it seems hard to reach the bound mentioned in Conjecture \ref{conj:1}.

\subsection*{Acknowledgement} The author would like to thank Nicolas Trotignon for his help and useful discussion.

\bibliographystyle{abbrv}
\bibliography{ISK4}

\end{document}